\documentclass[sigconf]{aamas} 
\renewcommand\footnotetextcopyrightpermission[1]{} 
\usepackage{balance} 

\usepackage{booktabs} 
\usepackage{tikz} 
\newcommand*\circled[1]{\tikz[baseline=(char.base)]{
            \node[shape=circle,draw,inner sep=2pt] (char) {#1};}}

\usepackage[linesnumbered,ruled,vlined]{algorithm2e}

\usepackage{xcolor}
\usepackage{booktabs}
\usepackage{enumitem}
\usepackage{url}
\usepackage{multicol}
\usepackage{multirow}
\usepackage{arydshln}
\usepackage{adjustbox}
\usepackage{soul}
\usepackage{bbm}

\renewcommand{\paragraph}[1]{\smallskip\noindent\textbf{#1}}

\usepackage{pifont}
\newcommand{\cmark}{\ding{51}}%
\newcommand{\xmark}{\ding{55}}%

\DeclareMathOperator*{\argmax}{arg\,max}

\setcopyright{ifaamas}
\acmConference[AAMAS '24]{Proc.\@ of the 23rd International Conference
on Autonomous Agents and Multiagent Systems (AAMAS 2024)}{May 6 -- 10, 2024}
{Auckland, New Zealand}{N.~Alechina, V.~Dignum, M.~Dastani, J.S.~Sichman (eds.)}
\copyrightyear{2024}
\acmYear{2024}
\acmDOI{}
\acmPrice{}
\acmISBN{}

\usepackage[absolute,overlay]{textpos}
\title[Analyzing Crowdfunding of Public Projects Under Dynamic Beliefs]{Analyzing Crowdfunding of Public Projects Under Dynamic Beliefs}

\author{Sankarshan Damle}
\affiliation{
  \institution{IIIT, Hyderabad}
  \city{Hyderbad}
  \country{India}}
\email{sankarshan.damle@research.iiit.ac.in}

\author{Sujit Gujar}
\affiliation{
  \institution{IIIT, Hyderabad}
  \city{Hyderbad}
  \country{India}}
\email{sujit.gujar@iiit.ac.in}

\begin{abstract}
In the last decade, social planners have used crowdfunding to raise funds for public projects. As these public projects are non-excludable, the beneficiaries may free-ride. Thus, there is a need to design incentive mechanisms for such strategic agents to contribute to the project. The existing mechanisms, like PPR or PPRx, assume that the agent's beliefs about the project getting funded do not change over time, i.e., their beliefs are static. Researchers highlight that unless appropriately incentivized, the agents defer their contributions in static settings, leading to a ``race'' to contribute at the deadline. In this work, we model the evolution of agents' beliefs as a random walk. We study PPRx -- an existing mechanism for the static belief setting -- in this dynamic belief setting and refer to it as PPRx-DB for readability. We prove that in PPRx-DB, the project is funded at equilibrium. More significantly, we prove that under certain conditions on agent's belief evolution, agents will contribute as soon as they arrive at the mechanism. Thus, we believe that by incorporating dynamic belief evolution in analysis, the social planner can mitigate the concern of race conditions in many mechanisms.
\end{abstract}

\keywords{Civic Crowdfunding, Martingale Theory}

\newcommand{\BibTeX}{\rm B\kern-.05em{\sc i\kern-.025em b}\kern-.08em\TeX}

\newtheorem{definition}{Definition}

\newtheorem{theorem}{Theorem}

\newtheorem{observation}{Observation}

\begin{document}


\pagestyle{fancy}
\fancyhead{}



\maketitle 

\begin{textblock}{15}(0.35,1)
\centering
\noindent\small In the Proceedings of the 23\textsuperscript{rd} International Conference on Autonomous Agents and Multiagent Systems (AAMAS), 2024, as an Extended Abstract.
\end{textblock}

\section{Introduction}
The process of raising funds for \emph{public or private projects} through voluntary contributions is known as \emph{crowdfunding}.  As the contributors may be \textit{strategic agents}, researchers analyze crowdfunding game-theoretically \cite{Alaei:2016:DMC:2940716.2940777,strausz2017theory,pp2,ccc_aaai23}. This work focuses on the crowdfunding of public projects such as parks, libraries, and community services. 

\paragraph{Provision Point mechanism for Public projects (PPP)}. \citet{bagnoli1989provision} present the seminal approach for crowdfunding of public projects, which we refer to as PPP. In PPP, a project issuer (PI) sets up the project's crowdfunding by announcing a target threshold, known as the \emph{provision point}. PI seeks voluntary contributions from interested agents towards this project before a known \emph{deadline}. If the net contribution crosses the provision point by the deadline, PI funds the public project through them. If the target is not met, PI returns the contributions. 

\paragraph{PPP \& Free-riding.} As public projects are non-excludable, strategic agents in PPP may choose not to contribute and \emph{free-ride}. Moreover, PPP also admits several inefficient equilibria \cite{bagnoli1989provision,healy2006learning}. The primary challenge in crowdfunding of public projects is thus the lack of \textit{incentives} for strategic agents to contribute. \emph{Provision Point mechanism with Refunds} (PPR)~\cite{zubrickas2014provision} addresses this challenge with the introduction of \emph{refund bonus schemes}. 

\paragraph{PPR.} With PPR, if the project is not funded, the agents receive their contribution and an additional refund proportional to their contribution. Significantly, the incentive structure avoids free-riding by incentivizing the agents to contribute. PPR also overcomes inefficient equilibria as \citet{zubrickas2014provision} proves that the total contribution equals the provision point at equilibrium. Subsequent works \cite{chandra2016crowdfunding,damle2018designing} build on PPR by introducing other refund schemes.

\paragraph{Modes of Crowdfunding.} The following two settings are possible for a project's crowdfunding. (i) \emph{Offline}: in which the participating agents are not aware of the history of the contributions and the net contribution at any epoch. (ii) \emph{Online}: where the net and the history of contributions are visible to each participating agent (e.g., online platforms like \url{kickstarter.com} and \url{spacehive.com}). We refer to crowdfunding over online settings as \emph{sequential crowdfunding}.

Particularly for sequential crowdfunding, \emph{blockchain}-based online platforms are becoming popular. Blockchain is an immutable, decentralized, and public ledger~\cite{nakamoto2019bitcoin}. These ledgers allow for pseudo-anonymous, transparent, and verifiable payments while eliminating the middle person. In practice, 
crowdfunding is now being deployed as \emph{smart contracts} over publicly distributed ledgers such as the \emph{Ethereum blockchain}~(e.g., \url{weifund.io} and \url{starbase.co}). Carrying out transactions through incurs \emph{gas} (a form of payment), and thus, there is a need to design efficient mechanisms for sequential crowdfunding. \citet{damle2018designing} present several refund schemes and show these schemes consume fewer gas units, and therefore, the corresponding crowdfunding mechanisms are \textit{efficient} to deploy as smart contracts over blockchains.

For an offline setting, PPR is an excellent choice. However, PPR induces a simultaneous game~\cite{chandra2016crowdfunding}. In sequential crowdfunding, such a game result in the agents deferring their contribution until the deadline, which in turn may result in the project not getting funded \cite{chandra2016crowdfunding,cason2021early}, i.e., a ``race" condition (RC). \citet{chandra2016crowdfunding} introduce \emph{Provision Point mechanism with Securities} (PPS), which employs a temporal refund scheme to avoid the race condition. Furthermore, \citet{damle2018designing} study various aspects of refund schemes to avoid the race condition and for efficient deployment in blockchain-based online settings.

\paragraph{Information Structure~\cite{damle2019ijcai}.} In addition to the above, crowdfunding of public projects also depends on the information available to the participating agents. To capture this, we define the tuple consisting of each agent's (i) valuation and (ii) belief as its \emph{information structure}. The existing literature majorly assumes that each agent is interested in the funding of the public project, i.e., its \emph{valuation} towards the project's funding is non-negative. Additionally, the literature also assumes that each agent has \emph{symmetric} belief, i.e., each agent believes that the public project will be funded with probability $1/2$ and not with $1/2$. Note that in the real world, the beliefs may be asymmetric. \citet{damle2019ijcai} present PPRx (which leverages PPR) for public projects when information structure allows positive valuation with asymmetric, yet \emph{static}, beliefs.

With this background, we highlight a few observations that need to be addressed in crowdfunding mechanism design.

\begin{observation}[\cite{kickstarter-blog}]\label{obs1}
Empirically, the probability of funding a project \emph{decreases} with an \emph{increase} in its duration.
\end{observation}

\begin{observation}[\cite{cason2021early}]\label{obs2}
Empirically, agents prefer to contribute even in the absence of refunds (refer to Figure~\ref{fig:belief} (left), Section~\ref{sec::pprxdb}).
\end{observation}

\paragraph{Belief Evolution.} The above observations point to a change in the agent's belief regarding funding the public project. For instance, with Observation~\ref{obs1}, we see that agents become reluctant to fund projects whose target deadlines are greater. Moreover, from Observation~\ref{obs2}, it is natural to assume that the availability of critical information, such as net contribution and the remaining time, will also impact the agent's belief.

We refer to such evolving beliefs as \emph{dynamic} beliefs. We model this belief evolution as a \emph{random walk}. We argue that each agent's step size, at any epoch, will be a posterior update depending on its prior belief and other auxiliary information (e.g., net contribution or the time remaining).

\paragraph{PPRx-DB.} This work primarily incorporates dynamic beliefs to analyze incentive-based civic crowdfunding mechanisms. To the best of our knowledge, PPRx is the only mechanism that incorporates asymmetric but static beliefs. We study PPRx under dynamic beliefs, and to distinguish our setting, we refer to it as \emph{Provision Point Mechanism for agents With Dynamic Belief} (PPRx-DB). We argue that agents' belief evolution will be a random walk. We identify conditions on the random walk under which we can characterize the sub-game perfect equilibria of the sequential game induced by PPRx-DB. In particular, by utilizing the evolution of each agent's random walk as a martingale, super/sub-martingale, we identify conditions wherein agents are naturally incentivized to contribute as soon as they arrive (i.e., avoid the race condition). Thus, though theoretically sound, complex mechanisms such as PPS may not be warranted in practice (see Table~\ref{tab:RL}).

%
\begin{table*}[!t]
\begin{small}
\centering
    \begin{tabular}{cccccc}
    \toprule
    \multirow{2}{*}{\textbf{Mechanism}}     & \multicolumn{2}{c
    }{\textbf{Information Structure}} & \multirow{2}{*}{\textbf{Setting}} & \multirow{2}{*}{\textbf{On Blockchain}} & \multirow{2}{*}{\textbf{RC}}  \\
    \cline{2-3}
    &  \textbf{Valuation} & \textbf{Belief} &  &  &  \\
    \hline
    PPR \cite{zubrickas2014provision}     &  Positive & Static \& Symmetric & Offline   & - & \textcolor{red}{\cmark}\\
    PPS \cite{chandra2016crowdfunding}     &  Positive & Static \& Symmetric & Online   & Inefficient & \textcolor{black!30!green}{\xmark}\\
    PPRG \cite{damle2018designing}     &  Positive & Static \& Symmetric  & Online   & Efficient & \textcolor{black!30!green}{\xmark} \\
    PPRN \cite{damle2019ijcai} &  Positive \& Negative & Static \& Symmetric & Offline   & - & \textcolor{black!30!green}{\xmark} \\
    PPRx \cite{damle2019ijcai} &  Positive & Static \& \textit{Asymmetric}  & Offline   & - & \textcolor{red}{\cmark}\\
    \cite{pp1} & Binary ($\{0,1\}$) & Static \& Symmetric & Offline & - & \textcolor{red}{\cmark} \\
     \hdashline
   { PPRx-DB}  &   Positive &  {\textit{Dynamic} \& \textit{Asymmetric}}  &{ Online}  & {Efficient} & \textcolor{black!30!green}{\xmark}/\textcolor{red}{\cmark}$^\ddagger$\\
  \bottomrule
  \multicolumn{6}{l}{$\ddagger$: Depends on certain conditions on Agent Beliefs (see Table~\ref{tab:results} for details)}
    \end{tabular}
    \caption{Comparing existing literature with PPRx-DB. Here, ``\xmark'' denotes that the mechanism avoids the Race Condition (RC), and ``\cmark'' denotes otherwise.}
    \label{tab:RL}
\end{small}
\end{table*}
%

\section{Related Work}
Designing mechanisms for crowdfunding of public projects with provision points is an active research area \cite{bagnoli1989provision,marx2000dynamic,morgan2000financing,zubrickas2014provision,chandra2016crowdfunding,chandra2017referral,damle2019ijcai,ccc_aaai23,padala21}. Provision Point mechanism for Public projects (PPP) \cite{bagnoli1989provision} is the first mechanism in this class. In PPP, agents contribute voluntarily to the public project until a deadline. If the project is funded, the agents receive a positive payoff; if not, their contributions are returned. However, PPP suffers from free-riding and has several inefficient equilibria \cite{healy2006learning}. 

To overcome free-riding, \citet{zubrickas2014provision} proposes Provision Point mechanisms with Refunds (PPR). PPR introduces the concept of a \emph{refund bonus scheme} paid out to the contributing agents if the public project is not funded -- along with their contributions. More significantly, the equilibrium contributions in PPR are such that the project is funded at equilibrium \cite{zubrickas2014provision}. Subsequently, researchers propose several mechanisms based on different agent models and refund schemes. We summarize the existing literature for refund mechanisms for crowdfunding of public projects with Table~\ref{tab:RL}; where RC refers to the race condition. PPRx~\cite{damle2019ijcai} is the closest work to ours; however, PPRx assumes static beliefs.

\section{Preliminaries}
We now (i) state our crowdfunding model, (ii) provide relevant game-theoretic definitions, (iii) describe different information structures, (iv) summarize PPRx and (v) define martingales.

    \paragraph{Crowdfunding Model.} The sequential crowdfunding of a public project $P$ set up by a \emph{Project Issuer} (PI) involves the following steps.
    \begin{enumerate}[leftmargin=*, noitemsep]
        \item PI announces a public project $P$, seeking voluntary contributions towards it. The announcement also comprises of the \emph{target threshold} ($H_0$), the \emph{project deadline} ($T$), and the \emph{refund bonus scheme} ($R$) (if deployed). Let, $\mathbf{T}=\{1,\dots,T\}$ denote the set of discrete epochs. 
        \item After step (1), each interested agent $j$ arrives to the crowdfunding platform at $a_j\in \mathbf{T}$. Upon arrival, each agent $j$ observes its valuation ($\theta_j$) towards $P$. 
        \item An interested agent $i$ may contribute $x_i$ to the crowdfunding platform at an epoch of time $t_i\in \mathbf{T}$. Let $\mathbf{A}=\{1,\dots,n\}$ denote the set of all contributing agents, with $C_0=\sum_{i\in\mathbf{A}} x_i$ as the total contribution towards the public project; and $\vartheta=\sum_{i\in\mathbf{A}} \theta_i$ as the total valuation of all the contributing agents. Observe that, as we consider sequential crowdfunding and as the agents are strategic, their contributions may depend on $\theta_i$ \emph{as well as} the total contribution up till $t_i$, i.e., $C_{t_i}$.
        \item PI states that $P$ is funded only if the total contribution equals $H_0$ at any point before or at the deadline, i.e., $P$ is funded only if $\exists t\in \mathbf{T} \mbox{~s.t.~}C_{t}\geq H_0$. Otherwise, $P$ is not funded.
        \item In the event that $P$ is not funded, PI returns the contributions of all the contributing agents. Agents may also be eligible for additional refunds depending on the mechanism deployed.
    \end{enumerate}

\smallskip\noindent\textit{Assumptions.}
       As standard in crowdfunding literature, our work assumes that,
            \begin{itemize}[leftmargin=*, noitemsep]
                \item We assume that $\vartheta$ is public knowledge s.t. $\vartheta>H_0$~\cite{zubrickas2014provision,chandra2016crowdfunding,damle2019ijcai}. It implies that there is sufficient interest in the funding of the public project.
                \item Agents arrive at the crowdfunding platform sequentially over time and \emph{not} simultaneously and can contribute only once to the project~\cite{chandra2016crowdfunding,damle2018designing}. 
                \item Agents do not have \emph{any} information regarding the funding of the project other than knowing the total contribution at any epoch of time~\cite{chandra2016crowdfunding,damle2018designing}. Moreover, agents are also not aware of the number of (i) contributions made before their arrival;  and (ii) agents that are yet to arrive at any epoch of time.
            \end{itemize}

Steps 1-5 induce a game among the interested agents. To analyze this induced game, we next define Sub-game Perfect Equilibrium. Towards these, let $\psi_i=(x_i,t_i)$ denote the strategy of each agent $i\in \mathbf{A}$ with $\psi=(\psi_1,\dots,\psi_n)$ as the vector of their strategies. As standard in the game theory literature, we use the subscript $-i$ to denote vectors without agent $i$. Further, let the payoff derived by an agent $i$ with valuation $\theta_i$ and with the strategy profile, $\psi$ be $\pi_i(\theta_i;\psi)$.

\paragraph{Sub-game Perfect Equilibrium (SPE).}
In a sequential setting, the agents may observe the actions of other agents over time. Thus, we must characterize the strategy profile such that it is the \emph{best response} for every agent to follow it at any time during the project, i.e., for every sub-game induced. We refer to such a strategy profile as \emph{Sub-game Perfect Equilibrium} defined as follows,

      \begin{definition}[Sub-game Perfect Equilibrium (SPE)] A strategy profile $\psi^\star = (\psi_1^\star,\dots,\psi_n^\star)$ is said to be a sub-game perfect equilibrium if for every agent $i$, if it maximizes the payoff $\pi_i(\psi^\star_i,\psi^\star_{-i|h^{a_i}};\theta_i),$ i.e. $\forall i \in \mathbf{A}$,
            $$\pi_i(\psi^\star_i,\psi^\star_{-i|h^{a_i}};\theta_i) \geq \pi_i(\psi_i,\psi^\star_{-i|h^{a_i}};\theta_i) \ \forall\psi_i, \forall h^{t}, \forall \theta_i.$$
            \end{definition}

 Here, $h^t$ is the history of the game till time $t$, constituting the net contribution and $\psi^\star_{-i|h^{a_i}}$ indicates that the agents who arrive after $a_i$ follow the strategy specified by $\psi^\star_{-i}$. Informally it means that, at every stage of the game, irrespective of what has happened, it is \emph{Nash Equilibrium} to follow the SPE strategy. That is, dropping the dependence on the history, we note that $\psi^\star = (\psi_1^\star,\dots,\psi_n^\star)$ satisfies Pure-strategy Nash equilibrium when $\forall i \in \mathbf{A}$,
            $$\pi_i(\psi^\star_i,\psi^\star_{-i};\theta_i) \geq \pi_i(\psi_i,\psi^\star_{-i};\theta_i) \ \forall\psi_i, \forall \theta_i.$$


\paragraph{Agent Information Structure.}  Similar to \cite{damle2019ijcai}, we consider agents with positive valuation and  asymmetric beliefs towards the funding of the public project. However, unlike \cite{damle2019ijcai}, we \emph{also} consider that an agent's belief may dynamically evolve over time. To the best of our knowledge, our work is the first to analyze such an information structure. 

    \paragraph{PPRx.} To incorporate asymmetric beliefs, \citet{damle2019ijcai} introduce PPRx. PPRx consists of two phases, (i) \emph{Belief Phase} (BP) with budget $B_B$ and deadline $T_B$ wherein each agent $i$ arrives at $a_{i,1}$ to the crowdfunding platform and submits its (static) belief $b_i$ at $t_{i,1}$ towards the funding of the project, based on which it gets a reward $m_i$ (known as Belief Based Reward); followed by (ii) \emph{Contribution Phase} (CP) with budget $B_C$ and deadline $T_C$ $(T=T_C+T_B)$ wherein each agent $i$ arrives at $a_{i,2}$ to the crowdfunding platform and contributes $x_i$ at $t_{i,2}$ to the project. Thus, $\psi_i=(b_i,t_{i,1},x_i,t_{i,2})$ is each agent $i$'s strategy in PPRx. 

We create the following subsets for agents with ``high" and ``low" belief: $\mathbf{A}_H=\{i|\forall i \in \mathbf{A} \mbox{~s.t.~}b_i\geq 1/2\}$ and $\mathbf{A}_L=\{i|\forall i \in \mathbf{A} \mbox{~s.t.~}b_i<1/2\}$. Trivially, $\mathbf{A}=\mathbf{A}_H\sqcup \mathbf{A}_L$. Let $\mathbbm{1}_X$ be an indicator random variable s.t. $\mathbbm{1}_X=1$ if $X$ is true and $\mathbbm{1}_X=0$ otherwise. Then, the payoff structure in PPRx $\forall i \in \mathbf{A}_H,$
    \begin{equation}
	\label{eqn1::PPRx}
    \pi_i = \mathbbm{1}_{{C_0}\geq H_0}\cdot (\theta_i - x_i + m_i) + \mathbbm{1}_{{C_0}< H_0}\cdot \left(\frac{x_i}{C_0}\cdot B_C\right)
	\end{equation}
Likewise, $\forall i \in \mathbf{A}_L$,
	 \begin{equation}
	\label{eqn2::PPRx}
        \pi_i= \mathbbm{1}_{{C_0}\geq H_0}\cdot ( \theta_i - x_i) + \mathbbm{1}_{{C_0}< H_0}\cdot \left(\frac{x_i}{C_0}\cdot B_C + m_i\right)
	\end{equation}

\smallskip
\noindent\textit{\underline{Funded and Unfunded Payoffs}.} We denote the payoff for agent $i$ when $C_0\geq {H}_0$ as $\pi_i^F$ (i.e., funded) and when $C_0< {H}_0$ as $\pi_i^{UF}$ (i.e., unfunded). In PPRx,
\begin{itemize}[leftmargin=*,noitemsep]
    \item $\forall i \in \mathbf{A}_H$: 
    \begin{align*}
            \mathbb{E}[\pi_i^F(x_i,\cdot)]&=b_i\cdot\left(\theta_i-x_i+m_i\right)\\ \mathbb{E}[\pi_i^{UF}(x_i,\cdot)]&=(1-b_i)\cdot\left(\frac{x_i}{C_0}\cdot B_C\right)
    \end{align*}
    \item $\forall i \in \mathbf{A}_L$:
    \begin{align*}
            \mathbb{E}[\pi_i^F(x_i,\cdot)]&=b_i\cdot\left(\theta_i-x_i\right)\\ \mathbb{E}[\pi_i^{UF}(x_i,\cdot)]&=(1-b_i)\cdot\left(\frac{x_i}{C_0}\cdot B_C+m_i\right)
    \end{align*}
\end{itemize}

\smallskip
\noindent\textit{\underline{Position of BBR}.}
We now describe the intuition behind the different positions of $m_i$ in Eq.~\ref{eqn1::PPRx} and Eq.~\ref{eqn2::PPRx}. We first note that, in PPRx, only one set of agents, either $\mathbf{A}_H$ or $\mathbf{A}_L$, get the BBR reward. The agents in $\mathbf{A}_H$ get the reward when the project is funded, implying their belief regarding the project's funding was `correct.' Likewise, agents in $\mathbf{A}_L$ get the reward when the project is not funded in accordance with their lower belief. The impact of such a utility structure is that agents with a higher belief have a greater equilibrium contribution than agents with lower belief -- which is a desirable outcome.

\paragraph{Belief Based Reward (BBR).} Based on the belief submitted by each agent in the Belief Phase, PPRx gives each agent $i$ a reward $m_i$, as follows \cite{damle2019ijcai},
    
    \begin{equation}
     \label{eqn::BBR}
 	  m_i = 
 	  \begin{cases}
 	    \frac{w_i}{\sum_{j\in\mathbf{A}_H} w_j} \times B_B   & \forall i \in \mathbf{A}_H    \\
        \frac{w_i}{\sum_{j\in \mathbf{A}_L} w_j} \times B_B   & \forall i \in \mathbf{A}_L \\
	  \end{cases}
	\end{equation}
for $w_i = \frac{y_i}{\sum_{j} y_j} \ \forall j \in S_{t_i}$ where  $y_i$ is the score calculated by the RBTS mechanism \cite{witkowski2012robust} depending on the belief $b_i$, while $S_{t_i}$ is the set consisting of all the agents that have reported their belief till $t_i$. BBR is (i) \emph{incentive compatible} and (ii) is a decreasing function of time. However, PPRx does not incorporate the evolution of agents' beliefs over time and its dependence on the total contribution, i.e., does not incorporate dynamic beliefs. 

\paragraph{Martingale Theory.}
A \textit{martingale} is a sequence of random variables such that the next value is equal to the current value in expectation,  conditioned over all prior values. However, for several applications, one cannot always guarantee this equality. To analyze such scenarios, we interest ourselves in \emph{bounding} the expected values. Such a sequence corresponds to a \emph{super-martingale} or a \emph{sub-martingale}. Formally, consider a discrete sequence of random variables $X_0,X_1,\dots$ evolving over time. Such a collection of random variables is referred to as a \emph{stochastic process}, denoted by $\{X_t\}_{t\in\mathbf{T}}$. 

     \begin{definition}[Martingales~\cite{martingalepdf}]
    A stochastic process $\{X_t\}_{t\in\mathbf{T}}$ such that $\mathbb{E}[X_t]<\infty$, is a
    \begin{itemize}
        \item Martingale if $\mathbb{E}[X_{t+1}|X_0,\dots,X_t] = X_t$
        \item Sub-martingale if $\mathbb{E}[X_{t+1}|X_0,\dots,X_t] \geq X_t$; and
        \item Super-martingale if $\mathbb{E}[X_{t+1}|X_0,\dots,X_t] \leq X_t$.
    \end{itemize}
    \end{definition}

In mechanism design literature, martingale theory is popularly used to model the dynamic evolution of an agent's private information. For e.g., \citet{chawla16} model agent's dynamic valuation for a product (such as Netflix subscription) over time as a Martingale. \citet{balseiro2018dynamic} model agent's expected utility as a Martingale to design a dynamic auction.

\begin{figure}
    \centering
    \includegraphics[width=\columnwidth]{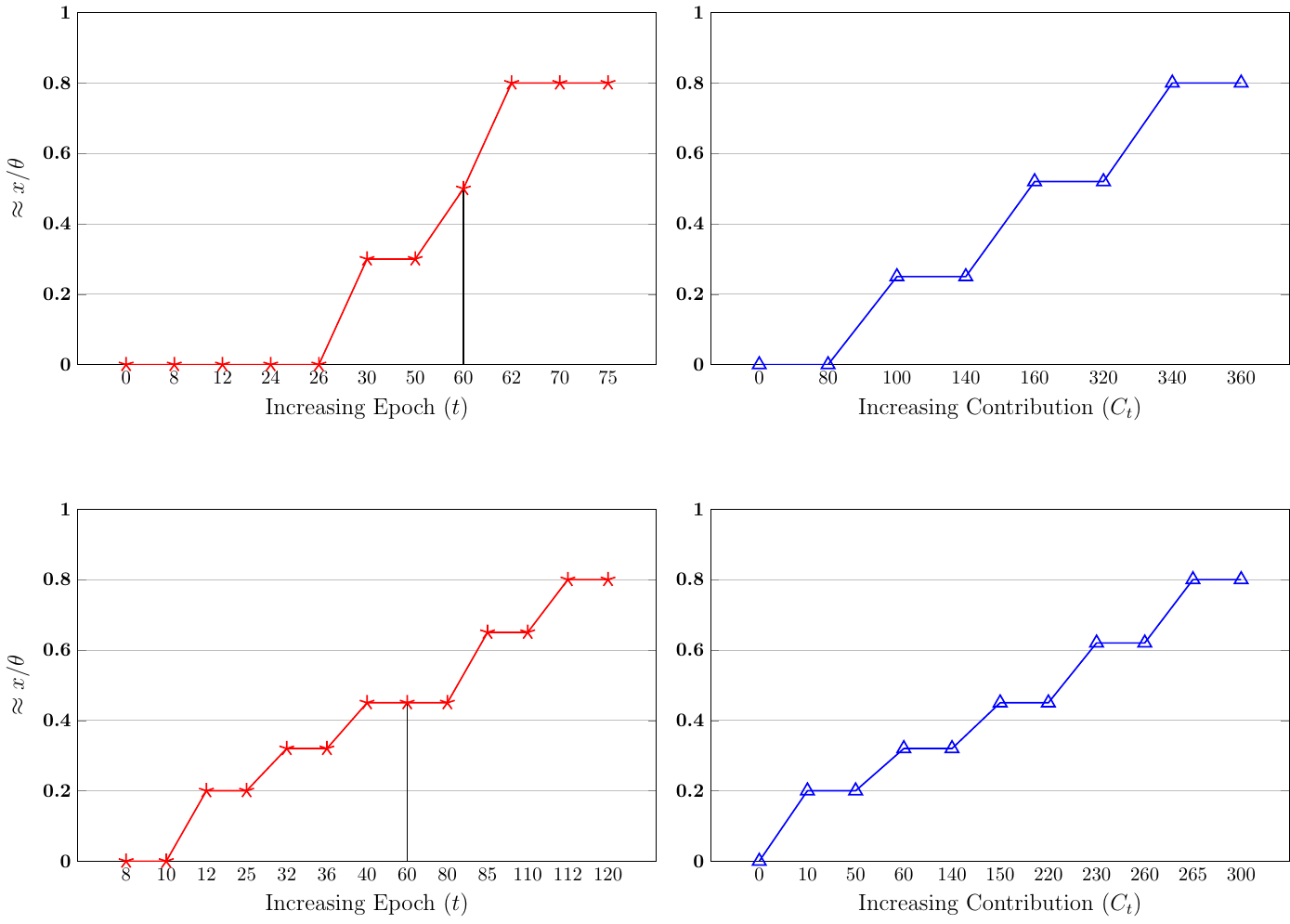}
    \caption{Plotting $\approx x/\theta$ for two randomly sampled agents using the dataset available with \cite{cason2021early}. The black vertical line in the left plots represents the end of the refund period. We observe that the agents contribute even post the refund stage, possibly implying a change in their beliefs. }
    \label{fig:belief}
\end{figure}

\section{PPRx-DB Mechanism\label{sec::pprxdb}}

This section analyzes PPRx with dynamic beliefs. We first present the agent's dynamic belief model. We then introduce PPRx-DB and provide agents' equilibrium contribution and the equilibrium time of contribution. We begin by presenting empirical evidence that agents' beliefs evolve during the crowdfunding process.

\paragraph{Observing Agent's Belief Evolution.} \citet{cason2021early} conduct real-world experiments to primarily test the impact of early refund bonus on a crowdfunding project's success. We use their data to provide the following insight regarding an agent's evolving belief.

 Figure~\ref{fig:belief} plots $x/\theta$, with varying $t$ and $C_t$, for two random agents from the dataset. Post $t>60$ seconds; the agents do not get refunds for their contributions. Yet, we observe that agents contribute post $t>60$ (Figure~\ref{fig:belief}(left)). Agents' contribution pattern also evolves with $C_t$ (Figure~\ref{fig:belief}(right)).

\subsection{Agent Dynamic Belief Model} 

We model the evolution of each agent's belief as a stochastic process over discrete epochs. Observe that this belief evolution may depend on available information at an epoch (e.g., net contribution). After each epoch, as an agent's belief can increase or decrease, we model it as a random walk.

For each agent $i\in\mathbf{A}$, let $\{b_{i,t}\}_{t\in\mathbf{T}}$ denote the random walk with $X_{i,t}$ as the random variable for the step size at an epoch $t$. More formally, let each agent $i$'s \textit{prior} belief regarding the project's funding be $b_{i,0}\in [0,1]$. At each epoch $t$, the agent's belief evolves in accordance with the available information, e.g., $C_{t}$ or remaining epochs $T-t$. 

Now, at each epoch $t\geq 1$, we denote agent $i$'s \textit{posterior} belief regarding the project's funding as: $b_{i,t}=b_{i,t-1}+X_{i,t}$. The sizes of the positive ($s_{i,+}()$) and negative ($s_{i,-}()$)steps (with ``$\circ$'' as auxiliary information) are:
\begin{equation*}
        X_{i,t}=
\begin{cases}
    s_{i,+}(C_t,T-t, \circ),& \text{with probability } p\in [0,1]\\
    s_{i,-}(C_t,T-t, \circ),  & \text{with probability } 1-p
\end{cases}
\end{equation*}

\noindent\textbf{Note 1.} $X_{i,t}$ captures the agent's belief evolution through the size of the step sizes, dependent on the available information. That is, an agent $i$'s belief evolves as $b_{i,t}\leftarrow b_{i,t-1} + X_{i,t}$ where $b_{i,0}$ is the agent's prior belief\footnote{Future work can explore the exact characterization of this random walk.}. We now have a model for the random walk, $\{b_{i,t}\}_{t\in\mathbf{T}}$. Our goal is to analytically derive equilibrium strategies for the agents conditioned on the behavior of $\{b_{i,t}\}_{t\in\mathbf{T}}$.


\subsection{PPRx-DB: Formal Description}
We refer to PPRx in such a setting as \emph{Provision Point Mechanisms for agents with Dynamic Belief} (PPRx-DB) to enhance readability. Protocol \ref{Pro:PPRx-DB} formally describes PPRx-DB.

We denote $\mathbf{T}_B=\{0,1,\dots,T_B\}$ and $\mathbf{T}_C=\{T_{B+1},\dots,T_B+T_C\}$, where $T=T_B+T_C$, as the sets constituting the discrete epochs of time for the two phases. For the analysis of the mechanism, let $\bar{\mathbf{T}}_C=\{1,\dots,T_C\}$ as the set of epochs of time for the CP (w.l.o.g). Consequently, the strategy for each agent $i\in \mathbf{A}$ becomes $\psi_i=(\hat{b}_i,t_{i,1},x_i,t_{i,2})$. The payoff structure of PPRx-DB is the same as in PPRx, i.e., Eqs. \ref{eqn1::PPRx} and \ref{eqn2::PPRx}. We next game-theoretically analyze PPRx-DB.

     \begin{algorithm}[!t]
	\DontPrintSemicolon\small
    \renewcommand{\algorithmcfname}{Protocol}
    \begin{itemize}[leftmargin=*]
    \item \textbf{Belief Phase (BP):}
  			\begin{enumerate}
            \item PI announces the start of the phase
            \item Agent $i$ enters at $a_{i,1}$ and submits $\hat{b}_i$ at $t_{i,1}\in \mathbb{T}_B$
            \item Each participating agent $i$ is told their BBR reward $m_i$ 
            \end{enumerate}
    \item \textbf{Contribution Phase (CP):}
       		\begin{enumerate}
            \item PI announces the start of the phase
            \item Agent $i$ enters at $a_{i,2}$ and submits $x_i$ at $t_{i,2}$
            \item The protocol continues until the target is reached
            \item PI announces end of phase if target is not reached at $T_C$
       		\end{enumerate}
      \item Refunds are distributed as per the outcome by PI according to Eqs. \ref{eqn1::PPRx} and \ref{eqn2::PPRx} for set of agents $\mathbf{A}_H$ and $\mathbf{A}_L$, respectively    \end{itemize}

    \caption{\label{Pro:PPRx-DB}\emph{PPRx-DB Mechanism}}
	\end{algorithm}
        \normalsize

 		   
\subsection{PPRx-DB: Theoretical Analysis}
In this subsection, we first discuss the funding of the public project at equilibrium. Second, we provide the upper bound of the agents' equilibrium contribution. Last, we present the equilibrium time of contribution for agents based on the underlying condition of the agent's belief evolution. 

\subsubsection{Project Status at Equilibrium}
 In PPRx-DB, the public project is funded at equilibrium. That is, at equilibrium, the total contribution \emph{equals} the provision point, i.e., $C_0=H_0$, when $\vartheta>H_0$. Consider the following lemma.
 
        \begin{lemma}\label{DB::eqb}
        In PPRx-DB, the public project is funded at equilibrium, i.e., $C_0=H_0$ if $\vartheta>H_0$ with $B_B,B_C>0$.
        \end{lemma}
        \begin{proof}
             From Eq.~\ref{eqn1::PPRx} and Eq.~\ref{eqn2::PPRx}, at equilibrium, $C_0<H_0$ cannot hold, since $\exists i\in\mathbf{A}_H$ with $x_i<\theta_i+m_i$ or $\exists i\in {\mathbf{A}}_L$ with $x_i<\theta_i$, at least, as $\vartheta>H_0$. Such an agent $i$ can obtain a greater refund bonus by marginally increasing its contribution $x_i$ as $B_C>0$. Moreover, if $C_0>H_0$, any contributing agent can increase its payoff by marginally decreasing its contribution. Thus, at equilibrium $C^0=H^0$ holds, i.e., the project is funded at equilibrium.       
        \end{proof}

\subsubsection{Equilibrium Contribution: Upper Bound} 

We now analyze the equilibrium contribution of each agent $i\in\mathbf{A}$ in PPRx-DB. As each agent $i\in\mathbf{A}$ submits its belief $\hat{b}_i$ in the BP, PI can categorize each agent $i$ to the sets $\mathbf{A}_H$ or $\mathbf{A}_L$. We next independently compute equilibrium contributions for the agents in $\mathbf{A}_H$ and $\mathbf{A}_L$, respectively.

\circled{1}~\paragraph{For Agents with High Belief.}
 Lemma \ref{lemma1-contri} presents the equilibrium contribution, $x_{i}^\star$, analysis for each agent $i\in\mathbf{A}_H$. For the proof, we solve for $x_{i}^\star$ such that the (expected) funded payoff is greater than equal to the (expected) unfunded payoff. This is because from Lemma~\ref{DB::eqb} we know that at equilibrium the contributions are such that $C_0=H_0$.
        \begin{lemma} \label{lemma1-contri}
        In PPRx-DB, for each $i\in\mathbf{A}_H$, its equilibrium contribution is 
        \begin{equation}\label{eqb-c-1}
        x_i^\star\leq\frac{{H}_0b_{i,{t_{i,{2^\star}}}}(\theta_i+m_i)}{B_C(1-b_{i,{t_{i,{2^\star}}}})+{H}_0b_{i,{t_{i,{2^\star}}}}},
        \end{equation} where $t_{i,{2^\star}}\in\bar{\mathbf{T}}_C$ is its time of contribution at equilibrium.
        \end{lemma}
        \begin{proof}
        Since at equilibrium  $C_0=H_0$, each agent $i$ will contribute such that its funded payoff is no less than the highest possible unfunded payoff.
        That is, we solve for $x_{i}^\star$ such that $\mathbb{E}[\pi_i^F]\geq \mathbb{E}[\pi_i^{UF}]$, for each $i\in\mathbf{A}_H$. That is,
        \begin{align*}
            b_{i,t_{i,2^\star}}\cdot(\theta_i-x_i^\star+m_i)\geq (1-b_{i,t_{i,2^\star}})\cdot \frac{x_i^\star}{H_0}\cdot B_C \\
             \implies x_i^\star\leq\frac{{H}_0b_{i,{t_{i,{2^\star}}}}(\theta_i+m_i)}{B_C(1-b_{i,{t_{i,{2^\star}}}})+{H}_0b_{i,{t_{i,{2^\star}}}}}
        \end{align*}
        This proves the lemma.
        \end{proof}

\circled{2}~\paragraph{For Agents with Low Belief.} Similar to our analysis for the set of agents in $\mathbf{A}_H$, Lemma \ref{lemma2-contri} presents the equilibrium contribution analysis of each agent $i\in\mathbf{A}_L$.

        \begin{lemma}\label{lemma2-contri}
        In PPRx-DB, for each $i\in\mathbf{A}_L$, its equilibrium contribution is 
        \begin{equation}\label{eqb-c-2}
        x_i^\star\leq \frac{{H}_0b_{i,{t_{i,{2^\star}}}}\theta_i+{H}_0m_i(1-b_{i,{t_{i,{2^\star}}}})}{B_C(1-b_{i,{t_{i,{2^\star}}}})+{H}_0b_{i,{t_{i,{2^\star}}}}},
        \end{equation} where $t_{i,{2^\star}}\in\bar{\mathbf{T}}_C$ is its time of contribution at equilibrium.
        \end{lemma}
        \begin{proof}
        Similar to Lemma~\ref{lemma1-contri}, we again solve for $x_{i}^\star$ such that $\mathbb{E}[\pi_i^F]\geq \mathbb{E}[\pi_i^{UF}]$, for each $i\in\mathbf{A}_H$. That is,
                \begin{align*}
            b_{i,t_{i,2^\star}}\cdot(\theta_i-x_i)\geq (1-b_{i,t_{i,2^\star}})\cdot\left( \frac{x_i^\star}{H_0}\cdot B_C+m_i\right) \\
             \implies x_i^\star\leq \frac{{H}_0b_{i,{t_{i,{2^\star}}}}\theta_i+{H}_0m_i(1-b_{i,{t_{i,{2^\star}}}})}{B_C(1-b_{i,{t_{i,{2^\star}}}})+{H}_0b_{i,{t_{i,{2^\star}}}}}
        \end{align*}
        This proves the lemma.
        \end{proof}

    \subsubsection{Time of Equilibrium Contribution}

Firstly, note that the refund bonus scheme in PPR (or PPRx) is independent of time. This induces a simultaneous-move game in PPR \cite{chandra2016crowdfunding} or PPRx~\cite{damle2019ijcai}. However, in PPRx-DB, the dynamic evolution of agents' belief towards the public project results in \emph{variable} expected payoff for each agent -- dependent on their belief at each epoch. Thus, unlike PPR and PPRx, PPRx-DB does not induce a simultaneous-move game and can be deployed in sequential settings.

The challenge remains to identify the time at which an agent will contribute to the public project. Recall that we denote the funded payoff for agent $i$ as $\pi_i^F$ and the unfunded payoff as $\pi_i^{UF}$. Now, the complete payoff structure for agent $i$ is,
    $$
    \pi_i(\cdot)=\mathbbm{1}_{C_0\geq {H}_0}\cdot\pi_i^{F}(\cdot) + \mathbbm{1}_{C_0<{H}_0}\cdot \pi_i^{UF}(\cdot).
    $$

At equilibrium, the {expected} funded payoff is equal to the {expected} not funded payoff (Lemmas \ref{lemma1-contri} and \ref{lemma2-contri}). Thus, we have $\mathbb{E}[\pi_i]=\mathbb{E}[\pi_i^{UF}]$, $\forall i$ at equilibrium. In PPRx-DB, from Eq.~\ref{eqn1::PPRx} and Eq.~\ref{eqn2::PPRx}, we also have
    $$
    \pi_i^{UF}(x_i) = \frac{x_i}{C_0}\cdot B_C + c,
    $$
where $c=0,\forall i \in \mathbf{A}_H$ and $c=m_i,\forall i\in\mathbf{A}_L$. 

Now, the equilibrium time of contribution $t_{i,2^\star},~\forall i\in \mathbf{A}$, can be calculated as:
$$
t_{i,2^\star} = \argmax_{t_{i,2}\in\mathbf{\bar{T}}_C} \mathbb{E}[\pi_i^{UF}(x_i^\star)].
$$

The subsequent results indeed derive $t_{i,2^\star}$ for the set of agents in $\mathbf{A}_H$ and $\mathbf{A}_L$. For these, we remark that when an agent $i$ arrives at the Contribution Phase (CP), its belief at that epoch is the same as its prior belief, i.e., $b_{i,a_{i,2}}=b_{i,0}$. This stems from the fact that the agent has yet to observe the available information for any meaningful belief update.

%
\begin{figure*}
    \centering
    \begin{minipage}{0.44\textwidth}\centering
    \includegraphics[width=\linewidth]{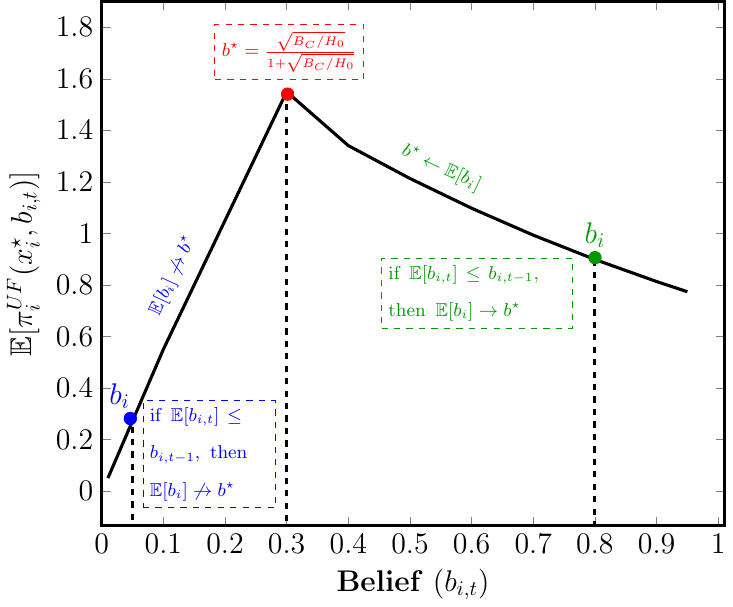}
    \caption*{(a) When $\{b_{i,t}\}_{t\in\bar{\mathbf{T}}_C}$ is a Super-martingale}
    \end{minipage}
    \begin{minipage}{0.44\textwidth}\centering
    \includegraphics[width=\linewidth]{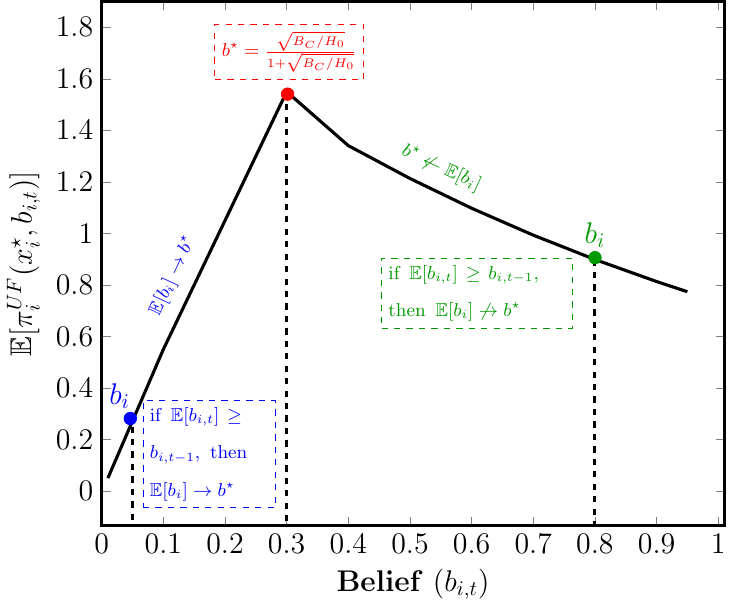}
    \caption*{(b) When $\{b_{i,t}\}_{t\in\bar{\mathbf{T}}_C}$ is a Sub-martingale}
    \end{minipage}
    \caption{Proof Intuition for Lemma~\ref{lemma_new_1}: Deriving Time of Equilibrium Contribution for Agent $i\in\mathbf{A}_H$}
    \label{fig:lemma4}
\end{figure*}

\circled{1}~\paragraph{For Agents with High Belief.} Consider the following lemma for each agent $i\in\mathbf{A}_L$.

\begin{lemma}\label{lemma_new_1}
In PPRx-DB, with $b^\star=\frac{\sqrt{B_C/H_0}}{1+\sqrt{B_C/H_0}}$, for each $i\in\mathbf{A}_H$, if:
\begin{enumerate}[leftmargin=*]
    \item $\{b_{i,t}\}_{t\in\bar{\mathbf{T}}_C}$ is a Martingale, then 
    $
    t_{i,2^\star} = T_C.
    $
\item $\{b_{i,t}\}_{t\in\bar{\mathbf{T}}_C}$ is a Super-martingale, then 
$$
    t_{i,2^\star} =
\begin{cases}
    a_{i,2}& \mbox{~if~} b_{i,0}\leq b^\star \\
    t\mbox{~s.t.~} b_{i,t}=b^\star& \mbox{~if~} b_{i,0}> b^\star
\end{cases}
$$
\item  $\{b_{i,t}\}_{t\in\bar{\mathbf{T}}_C}$ is a Sub-martingale, then 
$$
    t_{i,2^\star} =
\begin{cases}
    a_{i,2}& \mbox{~if~} b_{i,0}\geq b^\star \\
    t\mbox{~s.t.} b_{i,t}=b^\star& \mbox{~if~} b_{i,0}< b^\star
\end{cases}
$$
\end{enumerate}
\end{lemma}
\begin{proof} 
We broadly divide the proof into the following two steps: (i) Firstly, we show that $\mathbb{E}[\pi_i^{UF}(x_i^\star)]$ is increasing in $b_{i,t}$ only if $b_{i,t}\leq b^\star$. That is, $\mathbb{E}[\pi_i^{UF}(x_i^\star)]$ is maximized at $b^\star$. In (ii), to decide on $t_{i,2^\star}$, we condition on the underlying evolution of agent belief.  Consider the following.
\begin{itemize}[leftmargin=*]
    \item[(i)] \textbf{Deriving $b^\star$.} We have, 
    \begin{align*}
    \smaller
        \mathbb{E}[\pi_i^{UF}]=(1-b_{i,t})\cdot\frac{B_C}{H_0}\cdot \frac{{H}_0b_{i,{t_{i,{2^\star}}}}(\theta_i+m_i)}{B_C(1-b_{i,{t_{i,{2^\star}}}})+{H}_0b_{i,{t_{i,{2^\star}}}}}
    \end{align*}
Now, 
$$
\frac{\partial \mathbb{E}[\pi_i^{UF}]}{\partial b_{i,t}} >0 \iff \frac{b_{i,t}}{1-b_{i,t}}\leq \sqrt{\frac{B_c}{H_0}}.
$$
That is, $b_{i,t}\leq \frac{\sqrt{B_C/H_0}}{1+\sqrt{B_C/H_0}}=b^\star.$

\item[(ii)] \textbf{Deriving $t_{i,2^\star}$.} First, if  $\{b_{i,t}\}_{t\in\bar{\mathbf{T}}_C}$ is a Martingale, $\mathbb{E}[b_{i,t}]=b_{i,{t-1}}$. Thus, on expectation the value $ \mathbb{E}[\pi_i^{UF}]$ does not change. As such, agent $i$ has no incentive to contribute early, resulting in the race condition, i.e., $t_{i,2^\star}=T_C$. 

Second, if $\{b_{i,t}\}_{t\in\bar{\mathbf{T}}_C}$ is a Super-martingale, we have $\mathbb{E}[b_{i,t}]\leq b_{i,{t-1}}$. Since $\mathbb{E}[\pi_i^{UF}]$ increases till $b^\star$, if agents initial belief is less than $b^\star$, $b_{i,t}$ will not reach $b^\star$ (in expectation) implying agent $i$ must contribute as soon as it arrives, i.e., $t_{i,2^\star}=a_{i,2}$. However, if agent $i$'s initial belief is greater than $b^\star$ and since $\mathbb{E}[b_{i,t}]\leq b_{i,{t-1}}$, the agent waits till an epoch $t^\prime$ s.t. $b_{i,t^\prime}=b^\star$. Figure~\ref{fig:lemma4}(a) provides the proof intuition.

Last, if $\{b_{i,t}\}_{t\in\bar{\mathbf{T}}_C}$ is a Sub-martingale, we have $\mathbb{E}[b_{i,t}]\geq b_{i,{t-1}}$. Now, if agents initial belief is greater than $b^\star$, it is incentivized to contribute as soon as it arrives as its belief increases in expectation resulting in lesser $\mathbb{E}[\pi_i^{UF}]$. Likewise, if its initial belief is less than $b^\star$, than in expectation its belief will increase. That is, agent $i$ waits till an epoch $t^\prime$ s.t. $b_{i,t^\prime}=b^\star$. Figure~\ref{fig:lemma4}(b) provides the proof intuition.
\end{itemize}
This proves the lemma.
\end{proof}

\begin{table*}[t]
    \centering\small
    \begin{tabular}{ccccc}
    \toprule
\textbf{Agent Set}  &  $\{b_{i,t}\}_{t\in\bar{\mathbf{T}}_C}$  & $x_{i,t_{i,2^\star}}$  & $t_{i,2^\star}$ & \textbf{Race Condition} \\
\midrule
  \multirow{5}{*}{$\forall i \in \mathbf{A}_H$}      &  Martingale   & \multirow{5}{*}{$ \leq\frac{{H}_0b_{i,{t_{i,{2^\star}}}}(\theta_i+m_i)}{B_C(1-b_{i,{t_{i,{2^\star}}}})+{H}_0b_{i,{t_{i,{2^\star}}}}}$}     &  $T_C$ & \textcolor{red}{\cmark}  \\
  &  \multirow{2}{*}{Super-martingale} & & $a_{i,2} \mbox{~if~} b_{i,0}\leq b^\star$ & \textcolor{black!30!green}{\xmark}  \\
  &   &   & $  t\mbox{~s.t.} b_{i,t}=b^\star \mbox{~if~} b_{i,0}> b^\star$  & \textcolor{black!30!green}{\xmark} \\
  &  \multirow{2}{*}{Sub-martingale} & &  $a_{i,2} \mbox{~if~} b_{i,0}\geq b^\star$ & \textcolor{black!30!green}{\xmark}  \\
    &   &   & $  t\mbox{~s.t.} b_{i,t}=b^\star \mbox{~if~} b_{i,0}< b^\star$  & \textcolor{black!30!green}{\xmark} \\
\midrule
 \multirow{3}{*}{$\forall i \in \mathbf{A}_L$}      &  Martingale   & \multirow{3}{*}{$\leq \frac{{H}_0b_{i,{t_{i,{2^\star}}}}\theta_i+{H}_0m_i(1-b_{i,{t_{i,{2^\star}}}})}{B_C(1-b_{i,{t_{i,{2^\star}}}})+{H}_0b_{i,{t_{i,{2^\star}}}}}$} & $T_C$  & \textcolor{red}{\cmark} \\
 & Super-martingale & & $a_{i,2}$ & \textcolor{black!30!green}{\xmark}  \\
 & Sub-martingale & & $T_C$ & \textcolor{red}{\cmark}  \\
\bottomrule
    \end{tabular}
    \caption{Summary of Our Results for PPRx-DB. Here, ``\xmark" denotes that the mechanism avoids race condition.}
    \label{tab:results}
\end{table*}
%
\circled{2}~\paragraph{For Agents with Low Belief.} Similar to Lemma~\ref{lemma_new_1}, we now analytically present time of equilibrium contribution for agents in $\mathbf{A}_L$.

\begin{lemma}\label{lemma_new_2}
In PPRx-DB, with $\theta_i>m_i$ and $\theta_i<\frac{m_i\cdot H_0}{B_C}$ for each $i\in\mathbf{A}_L$, if
\begin{enumerate}
    \item $\{b_{i,t}\}_{t\in\bar{\mathbf{T}}_C}$ is a Martingale, then $t_{i,2^\star}=T_C$.
    \item $\{b_{i,t}\}_{t\in\bar{\mathbf{T}}_C}$ is a Super-martingale, then $t_{i,2^\star}=a_{i,2}$.
    \item $\{b_{i,t}\}_{t\in\bar{\mathbf{T}}_C}$ is a Sub-martingale, then $t_{i,2^\star}=T_C$.
\end{enumerate}
\end{lemma}
\begin{proof}
Similar to the proof for Lemma~\ref{lemma_new_1}, we broadly divide the proof in the following two steps: (i) Firstly, we show that $\mathbb{E}[\pi_i^{UF}(x_i^\star)]$ is increasing in $b_{i,t}$ if $\theta_i>m_i$ and $\theta_i<\frac{m_i\cdot H_0}{B_C}$. In (ii), to decide on $t_{i,2^\star}$, we condition on the underlying evolution of agent belief. Consider the following.

\begin{itemize}[leftmargin=*]
    \item[(i)] \textbf{$\mathbb{E}[\pi_i^{UF}]$ as an Increasing Function.} We first derive the condition in which $\mathbb{E}[\pi_i^{UF}]$ is an increasing function. We have,
    $$
    \mathbb{E}[\pi_i^{UF}] = (1-b_{i,t})\cdot\frac{B_C}{H_0}\cdot\frac{{H}_0b_{i,{t_{i,{2^\star}}}}\theta_i+{H}_0m_i(1-b_{i,{t_{i,{2^\star}}}})}{B_C(1-b_{i,{t_{i,{2^\star}}}})+{H}_0b_{i,{t_{i,{2^\star}}}}}
    $$
    Now,
    \begin{align}
&        \frac{\partial \mathbb{E}[\pi_i^{UF}]}{\partial b_{i,t}} >0 \iff (H_0-B_C)(m_i-\theta_i)b_{i,t}^2 + \nonumber  \\
        & 2B_C(m_i-\theta_i)b_{i,t}-(H_0+B_C)m_i+B_C\theta_i >0 \label{eqn::quad}
    \end{align}
    For $\mathbb{E}[\pi_i^{UF}]$ to be increasing, the quadratic in Eq.~\ref{eqn::quad} must be increasing. That is, its $\Delta<0$ and the first term must be positive. Through algebraic manipulations, we can show that these conditions will hold $\forall i \in \mathbf{A}_L$ iff $\theta_i>m_i$ and $\theta_i<\frac{m_i\cdot H_0}{B_C}$. 
    \item[(ii)] \textbf{Deriving $t_{i,2^\star}$.} As $\mathbb{E}[\pi_i^{UF}]$ is increasing in $b_{i,t}$ under $\theta_i>m_i$ and $\theta_i<\frac{m_i\cdot H_0}{B_C}$, we now derive $t_{i,2^\star}$ conditioned on the nature of the belief evolution. First, if $\{b_{i,t}\}_{t\in\bar{\mathbf{T}}_C}$ is a Martingale, then $\mathbb{E}[b_{i,t}]=b_{i,t}$. Trivially, the value $\mathbb{E}[\pi_i^{UF}]$ will not change in expectation for such a case. Thus, in practice, agent $i$ will defer its contribution to the deadline, i.e., $t_{i,2^\star}=T_C$.
    
    Second, if $\{b_{i,t}\}_{t\in\bar{\mathbf{T}}_C}$ is a Super-martingale, then $\mathbb{E}[b_{i,t}]\leq b_{i,t}$. In this case, a decrease in $b_{i,t}$ (in expectation) will imply a decrease in $\mathbb{E}[b_{i,t}]$. As such, agent $i$ will contribute as soon as it arrives, i.e.,  $t_{i,2^\star}=a_{i,2}$.
    
    Last, if  $\{b_{i,t}\}_{t\in\bar{\mathbf{T}}_C}$ is a Sub-martingale, then $\mathbb{E}[b_{i,t}]\geq b_{i,t}$. In this case, an increase in $b_{i,t}$ (in expectation) will imply an increase in $\mathbb{E}[b_{i,t}]$. As such, agent $i$ will defer its contribution till the deadline, i.e.,  $t_{i,2^\star}=T_C$.
    
\end{itemize}
This proves the lemma.
\end{proof}

\noindent\textbf{Note 2.} Table~\ref{tab:results} summarizes the results presented in this section. We analytically provide the equilibrium contribution and time of contribution based on the underlying property of agents' dynamic belief evolution. The equilibrium time of contribution also implies if the mechanism avoids the race condition or not. That is, when the equilibrium time of contribution equals the deadline, the race condition persists.

\subsection{PPRx-DB: SPE Strategy}
What remains to be shown is that the strategy, for each $i\in\mathbf{A}$, $\psi_i^\star = (b_i^\star,t_{i,1^\star},x_i^\star,t_{i,2^\star})$ where $b_i^\star=b_{i,0},t_{i,1^\star}=a_{i,1}$, $x_i^\star$ as defined in Lemma~\ref{lemma1-contri} and Lemma~\ref{lemma2-contri} and $t_{i,2^\star}$ as defined in Lemma~\ref{lemma_new_1} and Lemma~\ref{lemma_new_2} satisfies sub-game perfect equilibrium (SPE). To this end, consider the following theorem.

\begin{theorem}\label{thm:spe}
	For PPRx-DB, with the payoff structure as given by Eq. \ref{eqn1::PPRx} and Eq. \ref{eqn2::PPRx}, $\vartheta > H_0$ and $B_B,B_C>0$, we have $C_0=H_0$ and the set of strategies $\psi_i^\star = (b_i^\star,t_{i,1^\star},x_i^\star,t_{i,2^\star})$ where $b_i^\star=b_{i,0},t_{i,1^\star}=a_{i,1}$ and $(x_i^\star,t_{i,2^\star})$ as defined in Lemma~\ref{lemma1-contri} and Lemma~\ref{lemma_new_1} $\forall i\in\mathbf{A}_H$ and in Lemma~\ref{lemma2-contri} and Lemma~\ref{lemma_new_2} $\forall i\in\mathbf{A}_L$.
\end{theorem}
\begin{proof}
Firstly, from Lemma~\ref{DB::eqb}, we know that $C_0=H_0$ if $\vartheta > H_0$.
Next, $b_i^\star=b_{i,0}$ and $t_{i,1^\star}=a_{i,1}$ follows from the properties of BBR. More concretely, since BBR is incentive compatibility and decreasing with time, each agent $i\in\mathbf{A}$ reports its prior belief as soon as it arrive to the Belief Phase. Further, Lemmas~\ref{lemma1-contri},\ref{lemma2-contri} derive $x_i^\star$ and Lemmas~\ref{lemma_new_1},\ref{lemma_new_2} derive $t_{i,2^\star},\forall i\in\mathbf{A}_H$ and  $\forall i\in\mathbf{A}_L$, respectively.

The equilibrium strategy $\psi_i^\star$ depending on the aggregate contribution and current belief is also SPE. W.l.o.g., let agent $j$ arrive to the Contribution Phase (CP) last. If $C_{a_{j,2}}=H_0$, then its best response is contributing $x_{j,2^\star}=0$. If $H_0-C_{a_{j,2}}>0$, then irrespective of $H_0$ and $C_{a_{j,2}}$ its best strategy is $x_i^\star$ (defined in Lemmas~\ref{lemma1-contri},\ref{lemma2-contri}) and $t_{i,2^\star}$ (defined in \ref{lemma_new_1},\ref{lemma_new_2}). Using backward induction, we argue that it is the best response for every agent $i$ to follow its strategy $\psi_i^\star$, irrespective of history. That is, $\psi_i^\star$ satisfy SPE, $\forall i \in\mathbf{A}$.
\end{proof}
Theorem~\ref{thm:spe} present the SPE strategy for an agent. Without additional information/assumption regarding the belief evolution or future agents' contribution, we believe that these are a good starting point for mechanism design for crowdfunding of public projects with dynamic beliefs.

\section{Conclusion \& Future Work}
To the best of our knowledge, this paper is the first attempt at addressing the persistent issue of static beliefs in the existing literature on crowdfunding of public projects. Towards this, we model the dynamic belief update for each agent as a random walk. Empirical evidence available justifies this argument. Next, we analyzed PPRx with dynamic beliefs as PPRx-DB. We first derived the agent's equilibrium contribution as a function of their dynamic beliefs. In order to derive the time of equilibrium contribution, we condition the dynamic belief as a (i) Martingale, (ii) Super-martingale, and (iii) Sub-martingale. Based on these underlying conditions, we provide the time of equilibrium contribution. Consequently, we also showed the conditions at which PPRx-DB avoids the race condition. 

\smallskip
\noindent\textbf{Discussion \& Future Work.} Significantly, our results highlight that simpler mechanisms may also avoid the race condition, allowing a practitioner to save on-chain deployment costs. Future work can build on these results by (i) exploring other conditions that provide an analytical characterization of the agent's equilibrium time and contribution and (ii) empirically validating the evolution of the agent's dynamic belief as a martingale. In parallel, one can even attempt to learn an ML model for an agent's belief update.

\begin{acks}
    The authors would like to thank Prof. Timothy Cason for providing access to their dataset introduced in \cite{cason2021early}.  
\end{acks}

\bibliographystyle{ACM-Reference-Format} 
\bibliography{ref}

\end{document}